\documentclass[runningheads]{llncs}
\usepackage[utf8]{inputenc}
\usepackage[english]{babel}
\usepackage{verbatim}
\usepackage{amsmath, amssymb, amsfonts}
\usepackage{enumerate, enumitem}
\usepackage{array}
\usepackage{subfig}
\usepackage{tikz}
\usepackage{complexity}

\newcommand{\gopt}{G_{\!opt}}
\newcommand{\argmin}{\operatornamewithlimits{arg\,min\,}}

\newcommand{\pNe}{pairwise Nash equilibrium}

\newcommand{\Ne}{Nash equilibrium}
\newcommand{\ps}{pairwise stability}

\newcommand{\config}[1]{(G^{(#1)},\Gamma^{(#1)},\mathbf{P}^{(#1)})}
\newcommand{\dconfig}{(G,\Gamma,\mathbf{P})}

\newcommand{\rxtc}{\textsc{RX3C}}
\newcommand{\XTC}{\textsc{EXACT\ 3-COVER}}
\newcommand{\xtc}{\textsc{X3C}}
\newcommand{\LSC}{\textsc{LOWER\ SC}}
\newcommand{\lsc}{\textsc{LSC}}
\newcommand{\LSCE}{\textsc{LOWER\ SC\ EQUILIBRIUM}}
\newcommand{\lsce}{\textsc{LSCE}}

\newcommand{\is}{\textsc{IS}}
\newcommand{\BR}{\textsc{BEST\ RESPONSE}}
\newcommand{\br}{\textsc{BR}}

\begin{document}

\title{Network Formation: Heterogeneous Traffic, Bilateral Contracting and Myopic Dynamics}
\titlerunning{Heterogeneous Traffic, Bilateral Contracting and Myopic Dynamics}

\author{Carme \`Alvarez\thanks{Partially supported  MEC TIN-2007-66523 (FORMALISM).} and Aleix Fern\`andez}

\institute{\email{alvarez@lsi.upc.edu}  \, \, \email{aleixfdz@gmail.com} \\
ALBCOM Research Group. Universitat Polit\`ecnica de Catalunya.
}

\maketitle

\begin{abstract}
We study a network formation game where nodes wish to send traffic to other nodes. Nodes can contract bilaterally other nodes to form bidirectional links as well as  nodes can  break unilaterally contracts to eliminate the corresponding links.  Our model is an extension  of the model considered in~\cite{AJM:07,AJM:09}. The novelty  is that we do no require the traffic to  be  uniform all-to-all. Each node specifies the amount of traffic that it wants to  send  to any other node.
We characterize  stable topologies under a static point of view and we also study the game under a myopic dynamics.  We show its convergence to stable networks under some natural assumptions on the contracting functions. 
Finally we consider the efficiency of pairwise Nash topologies from a social point of view and we show that  the problem   of deciding the existence stable topologies of a given price is $\NP$-complete.

\end{abstract}

\section{Introduction}

Nowadays  social and economic networks, and even communication  networks are typically endogenous and  operate at scale that makes unpractical the use of centralized policies. In all these networks, nodes can be seen as  autonomous agents that wish to communicate to each other.  Each  agent may  choose whom to accept connections from and whom to connect to. Moreover  nodes can also decide  to break a set of non profitable  connections.    Network Formation Games provide a natural model with which to study networks that are formed by ad hoc and decentralized dynamics. 
A simple  exemple of network formation game is the model known as \emph{Local Connection Game} defined in~\cite{Fe:03}.

There are several studies related to network formation  games.  A seminal work is the paper~\cite{M:76} in which  graphs describe cooperative structures  and represent coalitions among the players. Subsequently in~\cite{AM:88} the authors consider dynamics for network formation (see~\cite{N:05} for a survey on  cooperative games).  In~\cite{Ae:04} the authors study a network design game known by~\emph{Fair Connection Game}. A weighted extension of this model is  considered in~\cite{CR:06}. Models in which nodes can create links unilaterally are known by unilateral connection games see e.g.~\cite{BG:00a,Se:03,HS:03,F:05,He:07}.  In the bilateral connection games  links are created only when both end nodes accept. Some examples of this type of games can be found in~\cite{JW:95,JG:00,FK:02} (See~\cite{J:03} for a survey).

 In this paper we study an extension of the network formation model with  bilateral contracting  defined in~\cite{AJM:07,AJM:09}.  Instead of considering only  an all-to-all communication pattern, we extend the model to a non uniform traffic pattern. 
We are interested in understanding  the behavior of the networks  when the agents interact to choose their connections. In particular we study  both the role of  the bilateral contracts and the role of  the heterogeneous traffic,  in the dynamic process of shaping a network.
  Following standard game theory we refer the decision makers as players.  The main elements of the game are the following:\\
 1. Nodes are players;\\
 2. Links represent bilateral agreements between their end nodes;\\
 3. Each node can  deviate its  strategy by breaking contracts with other nodes (unilateral deviation) or by agreeing to the creation    of a  link with another node (bilateral deviation);\\
4. Given a network topology the pay offs of nodes  depend on the cost of  participating in the network as well as on the payments between the end nodes of a link.

Our study  focuses on the stability of the networks,  not only from a static point of view, but also from a dynamic point of view.
Since in  our model players can deviate their strategies bilaterally we consider the natural notion of  \emph{pairwise Nash equilibrium } also used in~\cite{AJM:07,AJM:08,AJM:09}. 
Our first contribution is a general characterization of pairwise Nash stable topologies. We show that  these stable graphs are forests that  depend on the traffic matrix. If the traffic matrix corresponds to a uniform topology like all-to-all, then we obtain an equivalent characterization to the one shown in~\cite{AJM:07}.

In order to study how a network would evolve under the interactions between players, we define  a  discrete myopic dynamics.  We prove that  the computation of the Best Response   is $\NP$-hard and then, in order to avoid this computational hardness,  we define a  set of actions that can be taken by a node in one  round in such a way  that the computation of the best one can be done efficiently. We also show the convergence the dynamics to pairwise Nash stable configurations for any given traffic matrix under some natural assumptions on the contracting functions. 
Furthermore if the  traffic matrix correspond to trees or to complete graphs, then dynamics converges after a polynomial   number of rounds in expectation.

Our third contribution is related with the efficiency of the pairwise Nash topologies from both social point of view and computational point of view. For some particular traffic patterns we are able to characterize their optimal topologies an then we can discuss about the efficiency of their equilibria (how an equilibrium could be close to  the optimum). But from a computational point of view we show that  problem of deciding the existence  of a pairwise Nash equilibrium of a given social cost  is $\NP$-complete.

The rest of the paper is organized as follows. In section 2 we describe the game model. In section 3 we show a characterization of pairwise Nash topologies. In section 4 we define a natural myopic dynamics
and we prove its convergence to pairwise Nash topologies for any given traffic pattern. In section 5 we study the efficiency of equilibria from the social point of view as well as the computational complexity of deciding the existence of efficient equilibria. The proofs of many results are in online appendices.

\section{The Formation Game}

We consider a Network Formation Game where the players are nodes of a network that wish to connect with other nodes. This game is an extension of the model defined in~\cite{AJM:07,AJM:09} to heterogeneous traffic matrix. Our game models a scenario where  for each pair $i,j$ of nodes, 
 $i$ wants to send some amount of traffic  $t_{ij}\geq 0$ to  $j$  instead of considering the uniform all-to-all  traffic pattern studied in~\cite{AJM:07}.  
Nodes contract with other nodes bilaterally  to form bidirectional communication links.  Each contract is the result of a bilateral agreement, one node seeks the agreement and the other node accepts it, and then there is a transfer of utility from the seeking node to the accepting one.
Moreover each node has  a cost which depends on the resulting topology. The payment that each node has to make takes into account the link maintenance, the routing costs  and a disconnection cost. Given a network topology, nodes' payoffs are obtained from payments between nodes participating in the same link minus their  cost. 

We use the notation $G=(V,E)$ to represent a graph (undirected or directed). We denote by $ij$ the edge or undirected link between the nodes $i$ and $j$ and we denote by $(i,j)$ the arc or directed link from node $i$ to node $j$. By an abuse of notation we use the shorthand $ij \in G$ (or $(i,j) \in G$) instead of using the name $E$ of the set of edges $ij \in E$ (or arcs $(i,j) \in E$). Following the same style, we use $G+ij$ (or $G+(i,j)$) to represent the new graph resulting from adding $ij$ (or $(i,j)$ to the set $E$) and we use $G-ij$ (or $G-(i,j)$) to represent the new graph resulting from removing $ij$ (or $(i,j)$ to the set $E$). We denote by $C_u$ the connected component of $G$ that contains the node $u$. We refer to the number of nodes by $n$.

Formally,  the resulting network topology of our network formation game is an undirected graph $G=(V,E)$ where each node of $V$ represents a player and each edge $ij$ represents an undirected link between $i$ and $j$. Let $t_{ij} \geq 0$ the number of packets that $i$ wants to send to $j$. We represent this traffic pattern by a \textsl{traffic matrix} $T = (t_{ij})_{i,j \in V}$. Given a traffic matrix $T$ we define its associate undirected graph as $G_T = (V,\{ij : t_{ij} + t_{ji} > 0\})$. We assume that given network topology, traffic is routed along shortest paths. In case of multiples shortest paths of equal length, traffic is split equally among all available paths.

Each node $i \in V$ has a cost
$$C(i;G) = \pi\,\delta(i;G) + c_i\,f(i;G)+D(i;G)$$
where $\pi$ is the cost per link and $\delta(i;G)$ denotes the degree of $i$ in $G$, $c_i$ is the routing cost per unit of traffic and $f(i;G)$ is the total traffic that transits though $i$ and $D(i;G)$ is the \textsl{disconnection cost}. Note that the first term corresponds to the \textsl{link maintenance cost} and the second corresponds to the \textsl{routing cost}.

Let $Participants(i) = \{j : t_{ij}+ t_{ji} > 0\}$.
We assume that  the disconnection cost $D$ has to satisfy the following two  natural assumptions:\\
\emph{\textbf{Assumption 1} (About disconnection)}: $D(i;G)=0$ if $i$ is connected to all its participants. Otherwise, if $j\in Participants(i)$, and  and $C_i\not = C_j$, then
$D(i;G + ik)=D(i;G + il)$ when $k,l \in C_j$. The disconnection cost only depends on whether $i$ is connected to $j$ and not on the selected link $il$.\\
\emph{\textbf{Assumption 2} (About connection)}: Given a player $i$ and two configurations $(G,\Gamma,\mathbf{P})$ and $(G',\Gamma,\mathbf{P})$ such that $Participants(i) \cap C_i \subset Participants(i) \cap C'_i$ then $U(i;G,\mathbf{P}) < U(i;G',\mathbf{P}')$.

Let $p_{ij}$ denote a payment from $i$ to $j$. We assume that if a link $ij$ does not exists or if $i = j$, then $p_{ij} = 0$. We refer to $\mathbf{P} = (p_{ij})_{i,j \in V}$ as the \textsl{payment matrix}. This payment matrix will be specified later (it depends on a contracting function which at the same time depends on the resulting topology).

Then, the total utility of a node $i$ in graph $G$ with payment matrix $\mathbf{P}$ is defined by
\[ U(i;G,\mathbf{P}) = \sum_{j \neq i} (p_{ji} - p_{ij}) - C(i;G) \]

In order to formalize the agreements between the players we consider for each node $i$ two sets: $F_i \subseteq V \setminus \{i\}$ is the set of nodes that $i$ is willing to accept connections from and $T_i \subseteq V \setminus \{i\}$ is the set of nodes it wants to connect to. Let $\mathbf{T} = (T_i)_{i \in V}$ and $\mathbf{F} = (F_i)_{i \in V}$ be the strategy vectors. 

 Given strategy vectors $\mathbf{T}$ and $\mathbf{F}$, let $\Gamma = \Gamma(\mathbf{T},\mathbf{F})$ a directed graph representing the \textsl{contracting graph} which captures the direction of the link formation and let $G=G(\mathbf{T},\mathbf{F})$ the resulting topology.  A contract $(i,j) \in \Gamma$  if and only if $i \in F_j$ and $j \in T_i$.  An undirected link $ij \in G$ if and only if $(i,j)\in  \Gamma$ or $(j,i)\in \Gamma$

We assume that there is \textsl{contracting function} $Q(i,j;G)$ that specifies the transfer of benefit from $i$ to $j$ in when the topology is $G$. (If $Q(i,j;G) < 0$, then the transfer if from $j$ to $i$). Given strategy vectors $\mathbf{T}$ and $\mathbf{F}$, the \textsl{payment matrix} $\mathbf{P}=(p_{ij})_{i,j \in V}$ is defined as
\[
p_{ij} = \left\{ \begin{array}{l l}
  Q(i,j;G) & \mbox{if } (i,j) \in \Gamma \\
  0 & \mbox{otherwise}
\end{array} \right.
\]

We say that $\dconfig$ is a \textsl{feasible configuration} if there exists $(\mathbf{T},\mathbf{F})$ such that $G = G(\mathbf{T},\mathbf{F})$, $\Gamma = \Gamma(\mathbf{T},\mathbf{F})$ and $\mathbf{P} = \mathbf{P}(\mathbf{T},\mathbf{F})$. In this case we say that $(\mathbf{T},\mathbf{F})$ \textsl{generates} configuration $\dconfig$. Hence, given strategy vectors $(\mathbf{T},\mathbf{F})$, the utility of node $i$ is $U(i;G(\mathbf{T},\mathbf{F}),\mathbf{P}(\mathbf{T},\mathbf{F}))$. By an abuse of notation we will often use the shorthand $G$, $\Gamma$ and $\mathbf{P}$ to refer to specific instantiations of network topology $G(\mathbf{T},\mathbf{F})$, contracting graph $\Gamma(\mathbf{T},\mathbf{F})$, and payment matrix $\mathbf{P}(\mathbf{T},\mathbf{F})$, respectively.

The contracting functions can be interpreted as the outcome of a negotiation process which depends on the network topology. Instead of focusing on particular contracting functions we are interested in contracting functions having two natural properties: \emph{anti-symmetry} and \emph{affinity}.\\
\emph{\textbf{Property 1} (Anti-symmetry)}: $Q$ is \emph{anti-symmetric} if
$Q(i,j;G) = -Q(j,i;G)$. (The payment for a link  does not depend on which node asks for the connection).\\
\emph{\textbf{Property 2} (Affinity)}: $Q$ is affine if for each $i,j ,k$ such that $j \in Participant(i)$ and $k \not\in Participant(i)$ then $|Q(i,j;G + ij)| > |Q(i,k;G+ik)|$.


\section{Pairwise Nash Stability}

Our study  focuses on stability notion, so we need to use an appropriate equilibrium concept. In our model players may deviate their strategies unilaterally or bilaterally. In particular, a node $i$ may  choose to not create a link $ij$, not including  $j$ in $T_i$ or $F_i$. Even a node $i$ may break a link $ij$ by removing $j$ from $T_i$ or $F_i$. But node $i$ can not  create a new link $ij$ unilaterally, it depends on that $j$ accepts the contract from $i$. Formally, the link $ij$ is created if and only if 
$i\in T_j$ and $j\in F_i$,  or $j\in T_i$ and $i\in F_j$. If only node $i$ is asked to include $j$ in $T_i$ and it is the case that $F_j$ does not contain $i$,  then $i$ the link $ij$ can not be created.
Hence, in an stable situation no node $i$ has incentive to \emph{break} unilaterally any contract and no pair of nodes have incentive of creating bilaterally a link. This notion corresponds to the \emph{pairwise Nash stability} also used in~\cite{AJM:07}.
A \pNe{} can be defined as the conjunction of the classical \emph{ \Ne{}} and \emph{\ps{}} first introduced by~\cite{JW:95}. 

In order to define formally  the concept of \pNe{}, we first need to know how a deviation of the strategy vectors modifies the outcome configuration.

Given a strategy vector $(\mathbf{T}, \mathbf{F})$, let $(\mathbf{T}', \mathbf{F}')$ be the resulting strategy vector after applying a deviation on $(\mathbf{T}, \mathbf{F})$. Then the resulting configuration is defined as follows: $G' = G(\mathbf{T}', \mathbf{F}')$, $\Gamma' = \Gamma(\mathbf{T}', \mathbf{F}')$ and $\mathbf{P}' = (p'_{ij})_{i,j \in V}$ where
\[
p'_{kl} = \left\{
 \begin{array}{l l l}
  p_{kl}    & \mbox{if $(k,l) \in \Gamma'$ and $(k,l) \in \Gamma$} \\
  Q(k,l;G') & \mbox{if $(k,l) \in \Gamma'$ and $(k,l) \notin \Gamma$}\\
  0         & \text{otherwise}
 \end{array} \right.
\]

Let us remind the classical concepts of Nash equilibrium and pairwise stability.
A strategy vector $(\mathbf{T}, \mathbf{F})$ is a \emph{Nash equilibrium} if for every node $i$, any unilateral deviation  $T_i'$ and $F_i'$ does not increment its utility,  
$U(i;G, \mathbf{P}) \geq U(i;G', \mathbf{P}')$.
A configuration $\dconfig$ is \emph{pairwise stable} if and only if
for every  $uv \in G$, $U(u;G,\mathbf{P}) \geq U(u; G-uv,\mathbf{P}')$ (no node has incentive to delete a link), and\
for every $uv \not \in G$, if $U(u;G+uv,\mathbf{P'}) > U(u;G,\mathbf{P})$ then $U(v;G+uv,\mathbf{P'}) <U(v;G+uv,\mathbf{P}')$ (no pair  of nodes have both incentive to add a new link between them).\\

Now we have all the ingredients to define the appropriate notion of stability.

\begin{definition} A strategy vector $(\mathbf{T}, \mathbf{F})$  is  a \emph{pairwise Nash equilibrium} if $(\mathbf{T}, \mathbf{F})$ is a \emph{Nash equilibrium} and its  outcome  $\dconfig$ is  \emph{pairwise stable}. A configuration $\dconfig$ is \emph{pairwise Nash stable} if there exists a pairwise Nash equilibrium $(\mathbf{T}, \mathbf{F})$  such that $G= G(\mathbf{T}, \mathbf{F})$, $\Gamma= \Gamma(\mathbf{T}, \mathbf{F})$ and $\mathbf{P}= P(\mathbf{T}, \mathbf{F})$.
\end{definition}

In~\cite{AJM:07} and~\cite{AJM:09}, the  \pNe{} concept  was defined in a different way, but it is not hard to see that both definitions are equivalent.

In other words, in our game, \pNe{} means that no node $i$ has incentive to delete a set of contracts, and no pair of nodes $i,j$ agree in creating a new  link $ij$ ( i.e. if the utility of $i$ is incremented then the utility of $j$ is decremented). 

We are interested in characterizing the topologies that, independently of the other parameters, lead to a \pNe{} configuration. We say that a configuration $\dconfig$ is undirected if $(u,v) \in \Gamma$ then $(v,u) \notin \Gamma$. Notice that in  a no undirected configuration a contract may be broken without changing the topology. 

\begin{definition} \label{pnet}
A graph $G$ is a \pNe{} topology if and only if every feasible and undirected configuration $\dconfig$ is a \pNe{}.
\end{definition}

It is not hard to see that there exists \pNe{}  configurations $\dconfig$ where $G$ is not a topology \pNe{}. Note that by definition, the stability of a topology is independent of the contracting function.
In the following theorem we characterize pairwise Nash  stable topologies.

\begin{theorem} \label{tpnet}
Given a traffic matrix $T$, a graph $G$ is a \pNe{} topology if and only if $G$ satisfies the following properties:\\
1. For every pair $u,v \in G$ such that $u \in Participants(v)$ then $C_u = C_v$.\\
2. $G$ is a forest.\\
3. For every edge $uv \in G$ then, let be $G' = G - uv$, $C'_u \cap Participants(v) \neq \emptyset$ and $C'_v \cap Participants(u) \neq \emptyset$.
\end{theorem}
\begin{proof}
Let $G$ be a \pNe{} topology:\\
(i) Let us suppose that there exists a pair  $u,v$ such that $v \in Participants(u)$ and $C_u \neq C_v$. Then $u$ and $v$ can deviate bilaterally by adding the link $uv$ to increment their utilities.\\
(ii) If there exists a link  $uv \in G$  contained in a cycle then for  some contracting function we have that $p_{uv} - p_{vu} > 0$. Note that if $u$  eliminates $uv$ (breaking the contract $(u,v)$ or $(v,u)$) then $u$ still remains connected to its participants, the maintenance cost decreases by $\pi$, and the total traffic  through $u$ in $G-uv$ is not greater than that the one through $u$ in $G$.  Therefore  node $u$ can  increment its utility applying such unilateral deviation.\\
(iii) Let us suppose that  there are a node $u$ and link $uv$ such that $u$ remains connected to its participants in $G-uv$. Hence for a contracting function  such that $p_{uv} - p_{vu} > 0$, $u$ can increase its utility by  eliminating  the link $uv$.

Let $(G,\Gamma,\mathbf{P})$ be a feasible undirected configuration where $G$ satisfies  1,2 and 3. No node $u$ wants to break unilaterally a set of contracts $C \neq \emptyset$. If $u$ breaks a contract then by property (iii) $u$ becomes   disconnected  from some of its participants and then this unilateral deviation is not profitable.

Neither a pair  $u,v$ wants to deviate bilaterally.  Notice that if a link $uv$ is added, then for a contracting function such that $p_{uv} - p_{vu} > 0$, the  utility of $u$ is decreased strictly (the maintenance and routing costs of $u$ at least remain equal).
\qed
\end{proof}

Given  the above characterization it is not hard to see  that if a topology $G$ is a spanning forest of $G_T$ then $G$ is a \pNe{} topology: $G$ is acyclic, all nodes are connected to its participants and for each link $uv$, $u \in Participants(v)$, stronger statement than condition 3.

\begin{corollary} \label{cor1}
Let $T$ be a traffic matrix. Then every spanning forest of $G_T$ is a \pNe{} topology.
\end{corollary}
Notice  that there exists  \pNe{} topologies that are not spanning trees of $G_T$, figure~\ref{fig:ps} is an example. 
If we consider a uniform traffic then we obtain the following result which corresponds to the  characterization presented in~\cite{AJM:07} for the all-to-all case.

\begin{corollary} \label{cor2}
Let $T$ be a traffic matrix so that $G_T$ is a complete graph. Assuming that $\pi>0$, for any contracting function $Q$, a feasible undirected configuration $\dconfig$ is a \pNe{} iff $G$ is a tree.
\end{corollary}


\vspace*{-10pt}
\begin{figure}
\begin{center}
\subfloat{
  \begin{tikzpicture}
  \node at (0,0)  (a)  {$a$};
  \node at (0,2)  (b)  {$b$};
  \node at (2,0)  (c)  {$c$};
  \node at (2,2)  (d)  {$d$};
  \node at (-1,1)      {$G_T$};
  \draw (a) -- (c) ;
  \draw (a) -- (d) ;
  \draw (b) -- (c) ;
  \draw (b) -- (d) ;
  \end{tikzpicture}
  }
 \qquad\qquad
\subfloat{
  \begin{tikzpicture}
  \node at (0,0) (a)  {$a$};
  \node at (0,2) (b)  {$b$};
  \node at (2,0) (c)  {$c$};
  \node at (2,2) (d)  {$d$};
  \node at (3,1)      {$G$};
  \draw (a) -- (b) ;
  \draw (a) -- (c) ;
  \draw (b) -- (d) ;
  \end{tikzpicture}
}
\end{center}
\vspace*{-12pt}
\caption{$G$ \pNe{} no spanning tree of $G_T$}
\label{fig:ps}
\end{figure}


\vspace*{-10pt}

\section{Myopic Dynamics}

One of our main purposes is to define a myopic dynamics for our network formation game that converges to pairwise Nash equilibrium configurations. We are interested in  discrete myopic dynamics where  at each round players change  their strategies to maximize their current utility. 
Our dynamics  should allow to study the game as a process where players interact  over time through deviations. Furthermore, the fact of being 
myopic allows us to represent realistic situations where players update their strategic decisions having only a limited knowledge of the world. Instead of considering a long-term objective, at each round players take actions that can be seen as an efficient algorithm to optimize their current payoff immediately. 
This requirement precludes the use of the best response dynamics since in our network formation game to  find the deviation that maximize the utility of a node is computationally hard.

Formally, we define best response problem in the classic way as follows,

\noindent \BR{}: Given a tuple $\langle Q, T, \pi, (c_i)_{i \in G} \rangle$ that defines a game setting, a configuration  $(G, \Gamma, \mathbf{P})$, a node $u$,  and a value $C$
\emph{decide} whether there is  a  deviation for $u$ such that after applying it, $u$'s utility is greater or equal than $C$.


\begin{theorem} \label{br}
The \BR{} problem is $\NP$-complete.
\end{theorem}

This $\NP$-completeness result lead us to reconsider the definition of each round  of the dynamics. The main idea behind our dynamics definition is to restrict the best response  so that   the players' decisions can be computed in polynomial time and so that the dynamics converges to a pairwise Nash equilibrium.

Our dynamics is based in rounds. For every round $k$, let $\config{k}$  the configuration at the beginning of the round and $u_k$ the activated node by an \textsl{activation process}. An activation process is any discrete time stochastic process $\{U_k\}_{k \in \mathbb{N}}$ where $U_k$ are i.i.d. random variables from $V$ drawn with full support. A realization of an activation process is called an \textsl{activation sequence}

For the sake of simplicity we assume that $\pi > 0$ and we assume that any initial configuration $\config{0}$ is  a feasible undirected configuration.

At each  round  $k$ the activated node  $u_k$ can take one of the following actions:\\
1) to break a contract $(u_k,v)$ or $(u_k,v)$, therefore $\Gamma^{(k+1)} = \Gamma^{(k)} - (u_k,v)$ or $\Gamma^{(k+1)} = \Gamma^{(k)} - (v,u_k)$, respectively, $G^{(k+1)} =G^{(k)}-u_kv$ and $\mathbf{P}^{(k+1)} = \mathbf{P}^{(k')}$,\\
2) to update the payment of a contract $(u_k,v)$ or $(v,u_k)$, and then  $p^{(k+1)}_{u_kv} = Q(u_k,v;G^{(k)} + u_kv)$, $\Gamma^{(k+1)} = \Gamma^{(k)}$ and $G^{(k+1)} =G^k$,\\
3) to ask for a new contract $(u_k,v)$, and if it was accepted by both  $u_k$ and $v$, then  $\Gamma^{(k+1)} = \Gamma^{(k)} + (u_k,v)$,  $G^{(k+1)} =G^k+u_kv$  and  $\mathbf{P}^{(k+1)} = \mathbf{P}^{(k)'}$,\\
4) to do nothing, $\config{k+1}=\config{k}$.\\

>From all these possible actions, $u_k$ selects the one that maximizes its utility. If  $u_k$ has multiple optimal choices, we assume that one of them is selected randomly.

Notice that  by the definition of the myopic dynamics we have that any accessible configuration $\config{k}$ is an undirected configuration and satisfies that if  $(u,v)\not \in \Gamma^{(k)}$ then  $p^{(k)}_{uv}=0$. 

\begin{definition}[\emph{Convergence}]
Given any initial configuration $ \config{0}$  and an instance of the activation process, we say that the dynamics converges if there exists $K$ such that
$\config{k+1} = \config{k}$ for $ k > K$.\\
We say that, the dynamics converges uniformly if for every $\epsilon > 0$ there exists $K$ such that
$Pr \left[ \config{k+1} = \config{k}, \forall k > K \right] \geq 1-\epsilon$
where the probability is taken w.r.t. the activation process.
\end{definition}


We say that a configuration is a \emph{sink configuration} if whatever node is activated, it will choose to do nothing. 
Since we  are assuming that $\pi > 0$, then a sink configuration has to be an acyclic topology. Moreover, it is not hard to see that in any sink configuration no node can break a set of contracts (not only one contract) improving in this way its benefit. Neither,a pair of nodes can increment their benefit by adding a new link between them. Then we have the following stability property of the sink configurations.


\begin{proposition} \label{pro:sinkpne}
If $\config{k}$ is a sink configuration then $\config{k}$  is a pairwise Nash equilibrium.
\end{proposition}

The following theorem  establishes the convergence of our dynamics when the contracting function is antisymmetric and affine. From now on let us suppose that $Q$ is an antisymmetric and affine contracting function.

\begin{theorem}
Let $\config{0}$ be any initial configuration. For any activation process the dynamics converges uniformly. Further, if the activation process is uniform then the expected number of rounds is $O(e^n)$.
\end{theorem}

The main idea of the  proof is  based on defining a degree of configuration an then show that this degree decreases as the dynamics evolves. The so called degree  is defined by a tuple of four  no negative integer values. These values give us information about the relation between the current topology $G^{(k)}$ and the given  traffic matrix $T$. In fact we  prove that at each round,  some of these values never increase and if the current configuration is not a pairwise Nash equilibrium,  then these values decrease strictly  with a probability greater than zero. Finally,  since the configurations with minimum degree are stable, then we can conclude that our dynamics converges. Let us introduce formally  the key element of the proof and all the properties needed to prove the dynamics convergence.


\begin{definition}[\emph{Degree of a configuration}]
We associate to each configuration $\dconfig$ a degree denoted by  $\delta \dconfig$  and defined by 
$ \delta \dconfig =(C_F,C_E, A_E,A_P)$  where:
\begin{itemize}
\item[$C_F$] is the number of edges that have to be removed to have a spanning forest of $G$.
\item[$C_E$] is  the maximum number of edges, such that their contracts are agreed between two nodes no participants, in all spanning forests of $G$. $C_E = \max_{B \in ST(G)} |\{uv : uv \in E(B) \land u \notin Participants(v)\}|$, where $ST(G)$ is the set of all spanning forests of $G$.
\item[$A_E$] is the minimum number of edges to be  added in order to get a graph in which  each node is  connected to its  participants.
\item[$A_P$] is the number of no updated contracts, this is $|\{ (i,j) : (i,j) \in \Gamma, \ p_{ij} \neq Q(i,j;G)\}|$.
\end{itemize}

\end{definition}


Since any sink configuration $\dconfig$ is a pairwise Nash equilibrium, we have that is $\dconfig$ a feasible configuration, $G$ can not contain any cycle, and each node is connected to all its participants. 

\begin{lemma} \label{lem:sink}
If $\dconfig$ is sink configuration then we have $\delta \dconfig =(0,C_E,0,0)$. Furthermore, if $\delta \dconfig =(0,0,0,0)$ then $\dconfig $ is a sink configuration.
\end{lemma}
%

In order to simplify the notation, let $\delta \config{k}=(C_F^{(k)}, C_E^{(k)}, A_E^{(k)}, A_P^{(k)})$. 

\begin{lemma} \label{lem:dynamics:less}
$C_E^{(k+1)} + A_E^{(k+1)} \leq C_E^{(k)} + A_E^{(k)}$ and $C_F^{(k+1)} \leq C_F^{(k)}$.
\end{lemma}

\begin{proof}
Let us consider all the possible actions:\\
(a) If $u_k$ breaks a contract, then $C_F^{(k+1)} \leq C_F^{(k)}$ and  $C_E^{(k+1)} \leq C_E^{(k)}$. By breaking a contract it can occur that $A_E^{(k+1)}=A_E^{(k)}+1$, but in this case we have that $C_E^{(k+1)}=C_E^{(k)}-1$ and then $C_E^{(k+1)}+ A_E^{(k+1)} \leq C_E^{(k)} + A_E^{(k)}$.\\
(b) If $u_k$ updates one contract, then the topology remains the same.\\
(c) If $u_k$ asks whether $v$ accepts a contract $(u_k,v)$ w.l.o.g and $v$ accepts it. Then it can be the case that $C_{u_k}\not =C_v$. If  the  new link $u_kv$ connects  $u_k$ or $v$ to one of their (disconnected) participants  then $A_E^{(k+1)}=A_E^{(k)}-1$. If it is not the case, $u_k$ only has incentive to add $u_kv$ if  at least $- p_{u_kv} -\pi \geq 0$, and  $v$ has incentive to accept the contract only if at least   $p_{u_kv} - \pi \geq 0$. Since $\pi > 0$, it can not occur.  Finally, we can also  consider the case that  $C_{u_k}=C_v$, i.e. $u_kv$ adds a new cycle. But now we have the same conditions than in the previous case and then the  bilateral deviation does not occur.

Summarizing, $u_kv$ does not add a new cycle, $C_F^{(k+1)} = C_F^{(k)}$, $u_kv$  connects at least  two participants,  $A_E^{(k+1)} = A_E^{(k)}-1$, an since  only a new link is added, $C_E^{(k+1)} \leq C_E^{(k)}+1$.
\qed
\end{proof}

We can also guarantee that $C_E^{(k)} + A_E^{(k)}$ is decremented when $A_E^{(k)}=1$.

\begin{lemma}\label{lem:dynamic:equalone}
If $\delta \config{k}=(C_F, C_E, 1, A_P)$ then there exists a node $u$ such that its most profitable action is to create a new edge with one of its participants so that the new configuration is $(C_F, C_E, 0, A_P')$ where $A_P'\geq0$.
\end{lemma}

\begin{proof}
Since $A_E=1$, there exist $u,v$ such that $u\in Participants(v)$ and $C_u\not = C_v$. Let us suppose that $Q(u,v;G+uv)\leq 0$. There always exists such $u$ because of the antisymmetric property of $Q$. If $u$ is the active node, by the affinity property of $Q$ the most convenient for $u$ is to create a contract with $v$ than any other $w\not \in Participants(u)$.
\qed
\end{proof}

We are now able to prove the convergence of our dynamics. 
In order to analyze the expected number of rounds  to converge in the worst case, we assume that the activation process is uniform, i.e. for all nodes $u$, $ Pr\left[U_k = u\right] = \frac{1}{n}$, where $n$ is the number of nodes.  The convergence proof  continues to hold even the activation process is not uniform.

\begin{proof}[of theorem 2]

Let   $\delta \config{k}=(C_F, C_E, A_E, A_P)$  and let $\delta \config{k+1}=(C_F', C_E', A_E', A_P')$.
At each round $k$, depending of the current value $A_E$ we have one of the following cases:\\
$\mathit{A_E > 1}$: With probability $p\geq  \frac{A_E}{n}$ one new link will be added and with probability $p \leq 1 - \frac{A_E}{n}$ a contract can be broken.\\
$\mathit{A_E = 1}$: With probability $p\geq \frac{1}{n}$, by lemma \ref{lem:dynamic:equalone} the activated node  adds a link with a participant. With $p\leq \frac{n-1}{n}$ it can be added a link between no participants or if $C_F+C_E>0$ even a contract can be broken.\\
$\mathit{A_E = 0}$: By lemma~\ref{lem:dynamics:less} no link will be added. If a contract is eliminated then it can be the case that the corresponding link  does not belong to a cycle (cut link) and either it connects two participants  ($C_E'=C_E-1, A_E'=1$), or not  ($C_E'=C_E-1, A_E'=0$ and then $C_E' +A_E'<C_E +A_E$). Or it can be the case that the corresponding link is in a cycle ($C_F'=C_F-1, A_E'=0$).

Now we are going to analyze the expected number of rounds to reach a sink configuration. First of all note that $C_E+A_E\leq n$ and $C_F\leq n-1$.

Given a round $k$ let's show that in expectation, after $O(\frac{n^n}{n!})$ rounds either the dynamics reaches to a sink configuration or it arrives at least one time to an $A_E = 1$ configuration:\\
- if $A_E > 1$ then with probability $p\geq \frac{A_E!}{n^{A_E}}$, after $A_E-1$ rounds it shall end up to a configuration with $A_E = 1$. Hence in an expected number of rounds $O(\frac{n^n}{n!})$ we arrive at least one time at a $A_E=1$ configuration.\\
- if $A_E = 0$ after $O(n^2 \, \ln n)$ rounds either the topology has changed or we have a sink. In $O(n^2 \, \ln n)$ rounds, in expectation, all nodes are activated $n$ times, thus if topology remains equal we arrive to a feasible configuration, because all contracts are updated. With an additional $O(n \, \ln n)$ rounds all nodes are activated once more, and if topology has not changed then the configuration is a sink, because all nodes have chosen to do nothing. If the topology change it is because a contract has been broken. In this case if the corresponding link is a cutting edge that routes traffic then $A'_E = 1$ configuration. Otherwise $A'_E = 0$ configuration and $C_F'+C_E'<C_F+C_E$. Hence after $O(n^4 \, \ln n)$ rounds either the dynamics reaches to $A'_E = 1$ or to a sink (it is not possible change more than $n^2$ times the topology without reaching to an $A_E = 1$ configuration).

If we arrive $O(n)$ times at a $A_E = 1$ configuration then, in expectation, a node adds a link with a participant. 
So given a round $k$, in expectation, after $O(n^2 \frac{n^n}{n!})$ rounds the dynamics achieve a $C_E + A_E = 0$ configuration or a sink configuration.

Once the current configuration has $C_E + A_E = 0$, by lemma~\ref{lem:dynamics:less} never we will arrive to an $A_E > 0$ configuration, so after $O(n^4 \, \ln n)$ rounds the dynamics converges to an $(0,0,0,0)$ configuration.

\qed
\end{proof}

We notice that the assumptions made  on the contracting  function are not very strong and it seems that since we do not have many restrictions the path leading to a stable configuration can be very long, as we have seen just above. In the following we focus on the study of the dynamics behavior on  particular cases. We show that the dynamics is so efficient in the case of all-to-all traffic pattern as well as in the case of tree traffic patterns. 

\begin{theorem} \label{thm:complete_convergence}
Let $T$ an  all-to-all traffic matrix.
For a uniform activation process, the dynamics converges uniformly in  $O(n^4 \ln n)$ rounds.
\end{theorem}

\begin{theorem}  \label{thm:tree_convergence}
Let $T$ traffic matrix such that $G_T$ is a tree .
For a uniform activation process, the dynamics converges uniformly in  $O(n^4 \ln n)$ rounds.
\end{theorem}

\section{Efficiency and Complexity}

We are interested in quantifying how bad are the effects of the competitive and selfish behavior of the players not only from the computational point of view, but also from the qualitative point of view.
Since we consider any traffic pattern, in general  we can show that there exist stable configurations  that are no sufficiently good from the social perspective. If we fix the traffic pattern, then we obtain  some better results.

The social cost $SC$ of a graph $G$  is defined  by $SC(G) =  \sum_{u \in G} C(u;G)$.
  Therefore the social cost depends only of the topology. Notice that contracts induce zero-sum transfers among nodes and do not affect global efficiency. Then considering the cost social as the sum of the costs or as the sum of the utilities is the same (excepting for the sign). Therefore the social cost depends only of the topology.

The optimal topology $\gopt$  is the most efficient topology from the social prospective. Formally, $ \gopt = \argmin_{G} SC(G)$.
The price of a topology $G$ measures how far $G$ is from the optimum and is defined by  $P(G) = \frac{SC(G)}{SC(\gopt)}$.
Finally, the \textsl{price of stability} ($PoS$) and the \textsl{price of anarchy} ($PoA$) are defined as the minimum and maximum price of a \pNe{} topology, respectively.

For some particular traffic patterns we are able to characterize their optimal topologies an then we can discuss about the efficiency of their equilibria.
If the pattern is a tree, that is  $G_T$ is a tree then there is only one \pNe{} topology, that corresponds to $G_T$.  Furthermore  $G_T= \gopt$  since it is connected, and any connected graph has a routing cost greater than $G_T$ and a maintenance cost greater or equal than $G_T$.

\begin{proposition} \label{price:tree}
If $G_T$ is a tree, then $PoS = PoA = 1$.
\end{proposition}

If $G_T$ is a complete graph then the price of the anarchy changes substantially.
For example, if we consider a all-to-all traffic patern, this is $G_T$ is complete and all traffic has the same weight, then a simple path maximizes the distances between all nodes, thus maximizes the social cost.
An example of efficient topology for the same traffic pattern is the star centered at a node with minimal routing cost.

In general, if $G_T$ is the complete graph, then the price of the star is less than $2$. The maintenance cost is $2(n-1)\pi$, less or equal than the maintenance cost of the optimal. And for each packet sent by $u$ and received by $v$, in the optimal topology arises a cost greater or equal than $(c_u+c_v) t_{uv}$, while in the star arises exactly $(c_u+c_v+2c_{min}) t_{uv}$. The star centered in the minimal routing cost node is a spanning tree of $G_T$, so it is \pNe{} topology. It also can be prove that always exists a path with $PoA$ of the order $O(n)$ (if we have a weighted traffic pattern, the order of simple path affects the social cost).

Notice that the optimal depends of $\pi$ since from social point of view it may be preferable to have cycles when $\pi$ is small enough w.r.t. the amount of traffic through the nodes and the routing costs. For example, with $G_T$ complete and $\pi < c_{min} t_{min}$ the complete graph is the optimal ($t_{min} = min_{i,j \in V} \{ t_{ij} \}$), from a global point of view there is always incentive to add a link between two nodes.

\begin{proposition} 
If $G_T$ is a complete graph, then $PoS \leq 2$ and $PoA = O(n)$.
\end{proposition}

From the computational point of view we show that it is really to decide whether, given the parameters that define the network formation game, there exists a pairwise Nash topology with a social cost less than or equal to a given cost. Even the problem of deciding whether given a game setting, there is a topology with social cost less than or equal to a given cost.
Formally these problems are defined as follows,\\
\noindent \LSCE{}: Given a tuple $\langle T, \pi, (c_i)_{i \in G}, C \rangle$, that defines the game setting
\emph{decide} whether there is a \pNe{} topology $G$ such that $SC(G) \leq C$.\\
\noindent \LSC{}: Given a tuple $\langle T, \pi, (c_i)_{i \in G}, C \rangle$, that defines a game setting
\emph{decide} whether there is a topology $G$ such that $SC(G) \leq C$.

\begin{theorem} \label{thm:com}
\LSC{} and \LSCE{} are $\NP$-complete problems.
\end{theorem}


\bibliography{bibliografia}

\newpage
\appendix
\section{Pairwise Nash Stability}
\begin{proof}[corollary~\ref{cor1}]
Let $G$ be a spanning forest of $G_T$, then (1) for every pair $u, v \in G$ such that $u \in Participants(v)$ then $C_u = C_v$, (2) $G$ is a forest and (3) for every link $uv \in G$, since $u \in Participant(v)$ then if $G' = G - uv$, we have that  $C'_u \cap Participants(v) \neq \emptyset$ and $C'_v \cap Participants(u) \neq \emptyset$.
\qed
\end{proof}

\begin{proof}[corollary~\ref{cor2}]
If $G$ is a tree then $G$ is \pNe{}, by corollary~\ref{cor1}.

If $(G,\Gamma,\mathbf{P})$ is a \pNe{} configuration, $G$ has no cycles because $\pi > 0$ and all nodes are connected to its participants, therefore $G$ is a tree.
\qed
\end{proof}

\section{Myopic Dynamics}
\begin{proof}[proposition~\ref{pro:sinkpne}]
Let $(G^{(k)}, \Gamma^{(k)}, \mathbf{P}^{(k)})$ be a sink, then:
\begin{enumerate}
\item No node wants to update a contract. Thus for all $(u,v) \in \Gamma^{(k)}$, $p_{uv} = Q(u,v;G^{(k)})$, so $(G^{(k)}, \Gamma^{(k)}, \mathbf{P}^{(k)})$ is a feasible configuration.
\item No node wants to break a  contract. Then 
$G^{(k)}$ is an acyclic graph, otherwise it can be show that a node $u$ will have incentive to break a contract. Let $uv$ be a link that is contained in some cycle. Without loss of generality, we assume that  $p_{uv} - p_{vu} \geq 0$, then $u$ has incentive to break the link:  in $G-uv$ $u$  remains connected to its participants, the total of routing traffic through $u$ at $G-uv$ is less or equal that at $G$ and finally increases its payoff, because $p_{uv} - p_{vu} + \pi > 0$.

Now, let us suppose that there is a node $u$ that has incentive to break a minimal set  $A$ of links, where $|A| > 2$. Let $v$ be a node such that $uv \in A$, and let $G' = G-uv$, clearly $C'_v \cap Participants(u) = \emptyset$. Since  $u$ has no incentive to break $uv$, then $c_u \, f(uv;G) + \pi < p_{vu} - p_{uv}$. If we consider the deviation of node $u$ corresponding to break a set of links $A' = A -\{uv\}$, if  $G'' = G - A'$, then $c_u \, f(uv;G'') + \pi \leq c_u \, f(uv;G) + \pi < p_{vu} - p_{uv}$. Therefore if $A$ is a profitable deviation,  then $A'$  is profitable too. Note that  once the $A'$ links have been broken, to break  $uv$ does not increase the $u$ utility's. Since $|A| > |A'|$, then we have a contradiction with the fact that $A$ is minimal.

So in a sink configuration all unilateral deviations leads to a utility loss.

\item No pair of  nodes can improve their playoff adding a new link. Then, in  $(G^{(k)}, \Gamma^{(k)}, \mathbf{P}^{(k)})$ no bilateral deviation is profitable.
\end{enumerate}
\qed
\end{proof}

\begin{proof}[ lemma~\ref{lem:sink}]
If $(G,\Gamma,\mathbf{P})$ is a sink then is a \pNe{}: $G$ is acyclic, thus $C_F = 0$, $(G, \Gamma, \mathbf{P})$ is a feasible configuration, thus $A_P = 0$, and all nodes are connected to their participants, thus $A_E = 0$.

If the  degree of a configuration  $(G,\Gamma,\mathbf{P})$  is  $(0,0,0,0)$, then the configuration is feasible, $G$ is a spanning tree of $G_T$ ($G$ is acyclic and for all $uv \in G, \, u \in Participants(v)$), and finally  any action of a node $u$  other than do nothing decrements strictly its utility.  Notice that there is no incentive to update a contract because the configuration  is feasible. Since  for all $(u,v) \in \Gamma$ we have that $ (v,u) \notin \Gamma$ then every  link is  a cutting edge between two participants. Finally, if   the utility of $u$ does not decrease when a link $uv$ is added, then $p_{uv} < 0$. Hence $v$ has no incentive to accept the proposal.
\qed
\end{proof}

\begin{proof}[theorem~\ref{thm:complete_convergence}]
Given any  all-to-all traffic matrix, we have that $C_E = 0$ in any feasible configuration. Then no cutting edge will ever be deleted. While the current topology is not connected,  the optimal action of any node is to add a link. 

Once $G$ is connected, the expected number of rounds needed to  delete a link is  $O(n^2 \, \ln n)$. Thus in $O(n^4 \, \ln n)$ rounds the topology is a tree.

Once $G$ is a tree only remains to update all contracts for arrive to a skin, that needs an additional $O(n^2 \, \ln n)$ rounds.
\qed
\end{proof}

\begin{proof}[theorem~\ref{thm:tree_convergence}]
Let $G_T$  be a tree. We are going to show that if $A_E \geq 1$ , then with probability greater than $\frac{1}{n}$ the active node $u_k$ will add a $(u_k,v)$ where  $v \in Participant(u_k)$.

Notice that  if $A_E \geq 1$ then  there exists a set of nodes $U = \{u_1,u_2,\ldots,u_l\}$  such that each $u_i$ is not connected to all its participants. For each $u_i$ let  $P_i$ the set of connected components $P_i =\{C_{u_{i1}},C_{u_{i2}},\ldots\}$ such that if $P \in P_i$ then $P \cap Participants(u_i) \neq \emptyset$.

For each node $u_i$ let us  denote by $v_i$ the node chosen by $u_i$ to create a link as its best optimal action$u_i$ and let  $p_i \in Participants(u_i) \cap C_{v_i}$.

We define a directed graph $G' = (U,\{(u_i,p_i) : u_i \in U)$. $G'$ has at least a cycle because every  node in $G'$ is the source of an edge. Furthermore,  every cycle of $G'$ has length $2$. Notice that if there is a cycle with length greater than two $w_1,w_2,w_3,\ldots,w_1$ then $T_{w_1w_2} + T_{w_2w_1},\, T_{w_2w_3} + T_{w_3w_2},\,\ldots,\, T_{w_{l'}w_1} + T_{w_1w_{l'}} > 0$,  which is a contradiction with the fact that $G_T$ is a tree.

Hence, if $A_E \geq 1$ there exist  $u_i,u_j$ such that $u_i \in Participants(u_j)$, and $u_i$ wants to connect to $C_{u_j}$
and $u_j$ wants to connect to $C_{u_i}$ . By the affinity of $Q$ we have tha if  the active node is  $u_i$ or $u_j$ then it will add a link one of its participants.

\vspace{2mm}
At each round $k$, depending on  the current degree, we have one of the following cases:\\
$A_E = 0, C_E = 0$. After deleting  all cycles and after updating all contracts  a sink configuration is reached (with $A_E = 0, C_E = 0$ it is not possible to add o remove a cutting edge that routes traffic).\\
$A_E = 0, C_E > 0$. Eventually either  a sink configuration is reached  or a cutting edge which routes traffic is deleted .\\
$A_E > 0$. With probability greater than $\frac{1}{n}$ the active node will add a link with a participant. Then dynamics reaches to a configuration such that $A'_E = A_E - 1$, $C'_E = C_E$ configuration. And with probability less than $1-\frac{1}{n}$ the active node will add a link with a no participant and then $A'_E = A_E - 1$, $C'_E = C_E+1$.\\

In the worst case dynamics reaches to $C_E = 0, A_E = 0$  after  adding $n-1$ links between participants. Then it is needed to visit  an  $A_E > 0$ configuration $O(n^2)$ times to have  $C_E = 0, A_E = 0$. 

Furthermore if $A_E > 0$, the active node adds a link, When $A_E = 0$ but the current configuration is not a sink   $A_E = 0, C_E > 0$,  all contracts are updated before that any change in the topology occurs in  an expected number of rounds $O(n^2\, \ln n)$.  In addition, all cycles are deleted at  $A_E = 0, C_E > 0$ configurations. After deleting a cycle an additional $O(n^2\, \ln n)$  rounds in expectation are needed  to  update all contracts.

The dynamics passes  $O(n^2)$ times from an $A_E = 1$ configuration to $O(n^2 \, \ln n)$  configurations with  $A_E = 0$.  Besides,  $O(n^4 \, \ln n)$ rounds in expectation are needed  to delete all cycles.
\qed
\end{proof}

\section{Efficiency and Complexity}
\begin{proof}[proposition~\ref{price:tree}]
Let us suppose that there is a \pNe{}   topology $G$ such that  $G\not = G_T$. We know that $G$ must be connected and acyclic. If there  exists a node $u$ such that $\delta(u;G) > \delta(u;G_T)$, since $u$ has $\delta(u;G_T)$ participants, then  $G$ can not be a \pNe{} topology. Since $G$ is a tree and for all $\delta(u;G) \leq \delta(u;G_T)$ then $\delta(u;G) = \delta(u;G_T)$.

Finally we prove that any link $uv$ of  $G_T$ is also a link of $G$. We define $G_0 = G_T$ and for $i \geq 0, \, G_{i+1} = G_{i} - Leafs(G_i)$. For each $i\geq 0$ and for each $uv \in G_i$  such that $\delta(u;G_i) = 1$, we have $uv \in G$. In the base case the leafs of $G$ and $G_T$ are the same, and for each leaf $u$ such that $uv \in G$ and $G$ is a \pNe{} then $u \in Participants(v)$.Therefore $uv \in G_T$.  Let us  assume that the hypothesis holds for $i-1$. If there is a $uv$ such that $u$ is a leaf in $G_i$ and $uv \notin G$, by the induction hypothesis and because $\delta(u;G) = \delta(u;G_T)$ exists a link $uw \in G$ such that $uw \notin G_T$. The node $u$ wants to keep $uw$ because it is contained in the path from $u$ to $v$. If $w$ has not incentive to remove the link it is necessary that for some $w' \in Participant(w)$ the path in $G$ between $w$ and $w'$ passes through $u$. Then we have that $ww' \in G_T$, but this link is incident to a leaf $w'$ in $G_j$ for some $j < i$, for hypothesis $ww' \in G$, so $G$ has a cycle $(u,w,w',\ldots,u)$, contradiction.
\qed
\end{proof}

\begin{proof}[theorem~\ref{br}]
We transform \is{} to \br{}. Let $\langle G,C \rangle$ be a \is{} instance, where $G$ is a graph and $0 \leq C \leq |V(G)|$. We construct a \br{} instance as follows:
\begin{enumerate}[label=$-$]
\item $S$ is a star graph with center at the $u$ node, where $V(S) = V(G) \cup \{u\}$ ($u \notin V$) and $E(S) = \{uv : v \in V(G)\}$.
\item $\Gamma = \{(v,u) : v \in V(G)\}$ is a contract graph that generates $S$ topology.
\item $\mathbf{P}$ is the update traffic matrix generated by $\Gamma$, if $(v,u) \in \Gamma$ then $p_{vw} = Q(v,w;G)$, otherwise $p_{vw} = 0$
\item $Q$ is some contract function such that for all $G$ then $Q(v,u;G) = 2$ and $Q(u,v;G) -2$.
\item $\pi = 1$ is the maintenance cost.
\item $c_w = 1$ for all $w \in V$, is the routing cost.
\end{enumerate}

Now we show that $\langle G, C \rangle \in \is \iff \langle S,\Gamma,\mathbf{P}, Q, T, \pi, (c_i)_{i \in S}, u, C \rangle \in \br$.

If $\langle G, C \rangle \in \is$ then there is an independent set $I \subseteq V(G)$ such that $|I| \geq C$. Let us consider the $u$ action corresponding to delete $I' = \{uv : v \!\notin\! I\}$  links (breaking the $\{(v,u): v \!\notin\! I \}$ contracts). The $u$ utility is $U(u;S-I',\mathbf{P}') = |I| \geq C$, since $u$ remains connected only to nodes  $I$ which  don't send traffic in $S-I'$. Hence  $u$ has a  maintenance cost $|I|$ and receives  utility $2|I|$ from the other nodes.

Let us suppose that $u$ applies a deviation such that produces an utility  is greater  than or equal to $C$. The deviation neither can  add  a new link, nor  can update a contract. Then  the only possible deviation  is to break a set of contracts. Let $D$ be the set of nodes such that $u$ remains connected after breaking  $D'$ links and $U(u;S-D',\mathbf{P}') \geq C$. $D$ is an independent set, otherwise $u$ routes at least $|V(S)|$ units of traffic , and $u$ utility is $U(u;S-D',\mathbf{P}') \leq -2|V(S)| +      D < 0 \leq C$. Since $D$ is an independent set, then $|D| \geq C$, otherwise $U(u;S-D',\mathbf{P}') < C$.
\qed
\end{proof}

In order to prove \lsce{} and \lsc{} $\NP$-completeness we first define an intermediate $\NP$-complete language. This problem is just a \XTC{} restriction version.

\begin{definition}
Given a tuple $\langle T, S \rangle$, where $T = \{\tau_1,\ldots,\tau_{3t}\}$ is a set of terminals and $S = \{\sigma_1,\ldots,\sigma_s\}$ a subset of 3-subsets   $\sigma_i = \{\tau_{i1},\tau_{i2},\tau_{i3}\}$ of $T$, such that $\forall \sigma,\sigma' \in S \ |\sigma \cap \sigma'| \leq 1$; \textsl{decide} whether there is an exact cover $C$ of $S$, that is, $\bigcup_{\sigma \in C}\sigma = T$ and $\forall \sigma,\sigma' \in C \ \sigma \cap \sigma' = \emptyset$;
\end{definition}

\begin{theorem}
\rxtc{} is $\NP$-complete.
\end{theorem}
\begin{proof}
We transform \xtc{} to \rxtc{}. Let $\langle T, S \rangle$ be a \xtc{} instance, we construct a $\langle T', S' \rangle$  \rxtc{} instance as follows:
\begin{align*}
& \quad N = \{\#_{ij}, \#'_{ij} : 1 \leq i \leq s, 1 \leq j \leq 3\} \\
& \quad N \cap T = \emptyset \\
& \quad T' = T \cup N \\
& \quad S' =  \begin{array}{l l} 
  & \big\{ \{\#_{ij},\#'_{ij}, \sigma{ij}\} : 1 \leq i \leq s, 1 \leq j \leq 3 \big\} \; \cup \\
  & \big\{ \{\#_{i1},\#_{i2},\#_{i3} \} , \{\#'_{i1}, \#'_{i2}, \#'_{i3} \} : 1 \leq i \leq s \big\} \\
\end{array}
\end{align*}
Note that the reduction is correct, for all $\sigma, \sigma' \!\in\! S'$ if $\sigma \neq \sigma'$ then $|\sigma \cap \sigma'| \leq 1$.
\qed
\end{proof}

Now we show that that $\langle T,S\rangle \!\in\! \xtc \iff \langle T', S' \rangle \!\in\! \rxtc$.

If there is an exact 3-cover $C \subseteq S$ for $\langle T, S \rangle$, then we can construct an exact 3-cover $C' \subseteq S'$ for $\langle T', S' \rangle$:
\[
1 \leq i \leq s \quad
\begin{array}{l l}
\{\#_{ij}, \#'_{ij}, \sigma_{ij}\} \!\in\! C', \quad \quad 1 \leq j \leq 3 & \quad \text{if } \sigma_i \!\in\! C \\
\{\#_{i1}, \#_{i2}, \#_{i3}\}, \{\#'_{i1},\#'_{i2},\#'_{i3}\} \!\in\! C' & \quad \text{if } \sigma_i \!\notin\! C \\
\end{array}
\]
$C'$ contains all $N$ terminals only once. Further, since $C$ is a solution, $C$ contains all $T$ terminals only once, so $C'$ also contains all $T$ terminals only once.

If exists a exact 3-cover $C' \subseteq S'$ for $\langle T', S' \rangle$, then we can construct an exact 3-cover $C \subseteq S$ for $\langle T,S \rangle$:
\[
1 \leq i \leq s \quad
\begin{array}{l l}
\sigma_i \!\notin\! C & \quad \text{if } \{\#_{i1}, \#_{i2}, \#_{i3}\} \!\in\! C' \\
\sigma_i \!\in\! C & \quad \text{if } \{\#_{i1}, \#_{i2}, \#_{i3}\} \!\notin\! C' \\
\end{array}
\]
$C'$ contains all $T'$ terminals only once, if $\{\#_{i1}, \#_{i2}, \#_{i3}\} \!\notin\! C'$, then for all $1 \leq j \leq 3$, $\{\#_{ij}, \#'_{ij}, \sigma{ij}\} \!\in\! C'$.

\begin{proof}[theorem~\ref{thm:com}]
We transform \rxtc{} to \lsce{}. Let $\langle T, S \rangle$ be a \rxtc{} instance, we construct a $\langle M, \pi, (c_i)_{i \in V(GM)}, C \rangle$ \lsce{} instance as follows:

{\newsavebox{\tempbox}
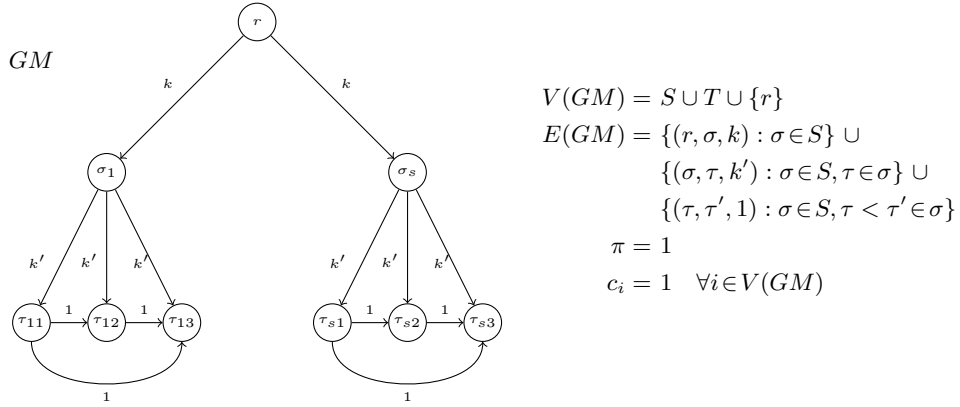
\begin{center}
\begin{figure}[!ht]
\sbox{\tempbox}{
\begin{tikzpicture}

 \node at (0,3.5) {$GM$};

 \node[draw=black,circle,inner sep=0pt, minimum size=5mm] at (3,4) (r) {\tiny{$r$}};

 \node[draw=black,circle,inner sep=0pt, minimum size=5mm] at (1,2) (s1) {\tiny{$\sigma_1$}};
 \node[draw=black,circle,inner sep=0pt, minimum size=5mm] at (5,2) (ss) {\tiny{$\sigma_s$}};

 \node[draw=black,circle,inner sep=0pt, minimum size=5mm] at (0,0) (t11) {\tiny{$\tau_{11}$}};
 \node[draw=black,circle,inner sep=0pt, minimum size=5mm] at (1,0) (t12) {\tiny{$\tau_{12}$}};
 \node[draw=black,circle,inner sep=0pt, minimum size=5mm] at (2,0) (t13) {\tiny{$\tau_{13}$}};

 \node[draw=black,circle,inner sep=0pt, minimum size=5mm] at (4,0) (ts1) {\tiny{$\tau_{s1}$}};
 \node[draw=black,circle,inner sep=0pt, minimum size=5mm] at (5,0) (ts2) {\tiny{$\tau_{s2}$}};
 \node[draw=black,circle,inner sep=0pt, minimum size=5mm] at (6,0) (ts3) {\tiny{$\tau_{s3}$}};

 \draw [->] (r) -- node[above left] {\tiny{$k$}} (s1);
 \draw [->] (r) -- node[above right] {\tiny{$k$}} (ss);

 \draw [->] (s1) -- node[near end, above left] {\tiny{$k'$}} (t11);
 \draw [->] (s1) -- node[near end, above left] {\tiny{$k'$}} (t12);
 \draw [->] (s1) -- node[near end, above left] {\tiny{$k'$}} (t13);

 \draw [->] (ss) -- node[near end, above left] {\tiny{$k'$}} (ts1);
 \draw [->] (ss) -- node[near end, above left] {\tiny{$k'$}} (ts2);
 \draw [->] (ss) -- node[near end, above left] {\tiny{$k'$}} (ts3);

 \draw [->] (t11) -- node[above] {\tiny{$1$}} (t12);
 \draw [->] (t12) -- node[above] {\tiny{$1$}} (t13);
 \draw [->] (t11) .. controls +(down:1cm) and +(down:1cm) .. node[below] {\tiny{$1$}} (t13);

 \draw [->] (ts1) -- node[above] {\tiny{$1$}} (ts2);
 \draw [->] (ts2) -- node[above] {\tiny{$1$}} (ts3);
 \draw [->] (ts1) .. controls +(down:1cm) and +(down:1cm) .. node[below] {\tiny{$1$}} (ts3);
\end{tikzpicture}
}
\subfloat{\usebox{\tempbox}}%
\hspace{-3cm}
\subfloat{\vbox to \ht\tempbox{%
  \vfil
\begin{flushleft}
{\small
\begin{align*}
  V(GM)      =\ & S \cup T \cup \{r\} \\
  E(GM)      =\ & \{(r,\sigma,k) : \sigma \!\in\! S\} \ \cup \\
           & \{(\sigma,\tau,k') : \sigma \!\in\! S, \tau \!\in\! \sigma\} \ \cup \\
           & \{(\tau,\tau',1) : \sigma \!\in\! S, \tau < \tau' \!\in\! \sigma\} \\
  \pi    =\ & 1 \\
  c_i    =\ & 1 \quad \forall i \!\in\! V(GM)
\\
\end{align*}
}
\end{flushleft}

  \vfil}}%
\end{figure}
\end{center}}

The traffic matrix $M$ is defined by the directed weighted graph $G$: if $(i,j,w) \in G$ then $m_{ij} = w$, otherwise $m_{ij} = 0$.

Given a possible 3-cover $S' \subseteq S$ for $\langle T, S \rangle$ we can construct the associated graph $G(S')$ where $V(G(S')) = V(GM)$ and
\[E(G(S')) = \{ (r,\sigma) : \sigma \!\in\! S \} \cup \{ (\sigma,\tau) : \sigma \!\in\! S', \tau \!\in\! \sigma\}\]
If $S'$ is an exact 3-cover then $G(S')$ is a spanning tree of $GM$. We define $C$ in such way that if $S'$ is a exact 3-cover then $C(G(S')) = C$.

Let $S'$ be an exact 3-cover, we denote by $C_{XY}$ the routing cost arising for the traffic sent from $X$ to $Y$ nodes, in the $G(S')$ topology.

In $G(S')$, $\forall \sigma \in S \ t_{r\sigma} = k$. For all $\sigma$ such that $(r, \sigma) \!\in\! E$, then the traffic between a $\sigma$ and $r$, only is routed only through $\sigma$ and $r$ nodes, therefore $C_{RS} = 2sk$.

In $G(S')$, $\forall \sigma \!\in\! S,\, \forall \tau \!\in\! \sigma$, $t_{\sigma\tau} = k'$. For a $\sigma \!\in\! S'$ node the shortest path to $\tau \!\in\! \sigma$ is $\langle \sigma,\tau \rangle$. Instead to a $\sigma \!\notin\! S'$ node the shortest path to $\tau \!\in\! \sigma$ is $\langle \sigma,r,\sigma',\tau \rangle$ where $\sigma' \!\in\! S'$. Therefore we have that $C_{ST} = 2k'3t + 9k'(s-t) = 6k'(3s-2t)$.

Finally to calculate $C_{TT}$ we distinguish again  two cases. Let $\tau_1, \tau_2 \!\in\! \sigma$ such that $\sigma \!\in\! S'$, the shortest path is $\langle \tau_1,\sigma,\tau_2 \rangle$. The total amount of traffic arises for these nodes is $12t$, because $|S'| = 3t$. If $\tau_1, \tau_2 \!\in\! \sigma$ such that $\sigma \!\notin\! S'$, the shortest path is $\langle \tau_1,\sigma,r,\sigma',\tau_2 \rangle$ where $\sigma' \!\in\! S'$ because $\forall \sigma, \sigma' \!\in\! S$ such that $\sigma \neq \sigma'$, $|\sigma \cap \sigma'| \leq 1$. $|S-S'| = s-t$ so the cost is $24(s-t)$, in total we have $C_{TT} = 4(6s-3t)$.

We  define $k$, $k'$ and $C$ from the $C_{XY}$ values in such a way  that only the topologies $B = G(S')$ for an exact 3-cover $S'$ satisfies $SC(B) \leq C$:
\begin{align*}
C_{TT} & = 4(6s-3t)                             &  C_{ST} & = 6k'(3s-2t) \\    
k'     & = C_{TT}                               &  k & = C_{ST} + k' \\
C_{RS} & = 2sk                                  &  C & = C_{RS} + C_{ST} + C_{TT} + 2 \pi (s+t)
\end{align*}

Now we show that $\langle T,S\rangle \!\in\! \rxtc \iff \langle M, \pi, (c_i)_{i \in V(GM)}, C \rangle \!\in\! \lsce$:

If $S' \subseteq S$ is an exact 3-cover of $T$, the tree $G(S')$ is a \pNe{} topology such that  $SC(G(S')) = C$.

Let  $G=(V,E)$ be a \pNe{} such that $\SC(G) \leq C$. $G$ is a tree and the maintenance cost is $2\pi(s+t)$.
Below we show that $G$ is a spanning tree of $GM$:
\begin{enumerate}[label=(\roman{*})]
\item $\forall \sigma\!\in\!S \ (r,\sigma)\!\in\!E$, otherwise we have a partition $\{X,Y\}$ of $S$ such that  $\sigma \!\in\! X \!\!\!\implies\!\!\! (r, \sigma) \!\in\! E$, $\sigma \!\in\! Y \!\!\!\implies\!\!\! (r, \sigma) \!\notin\! E$ and $|Y| \geq 1$. Then the routing cost arising for the traffic sent from $R$ to $S$ is greater or equal than $2k|X| + 4k|Y| = 2k(|X|+|Y|) + 2k|Y|$ and the social cost $\SC(G) > 2ks + 2k|Y| + 2\pi(s+t) \geq 2sk + 2k + 2\pi(s+t) > C$.
\item $\forall \tau\!\in\!T$ there is a unique $\sigma\!\in\!S$ such that $\tau\!\in\!\sigma$ and $(\sigma,\tau)\!\in\!E$, otherwise or exists at least a $\tau \!\in\! T$ with $\sigma,\sigma'$ such that $(\sigma, \tau),(\sigma',\tau) \!\in\! E$ (this is not possible because implies that $G$ has a cycle), or we have a partition $\{X,Y\}$ of $T$ such that $\tau \!\in\! X \!\!\!\implies\!\!\! \exists \sigma (\sigma, \tau) \!\in\! E \land \tau \!\in\! \sigma$, $\tau \!\in\! Y \!\!\!\implies\!\!\! (\forall \sigma \tau \!\in\! \sigma \!\!\!\implies\!\!\! (\sigma, \tau) \!\notin\! E)$ (it is possible that these $\tau$ are connected to a $\sigma$ node such that $\tau \!\notin\! \sigma$, to the $r$ node or to other $\tau' \!\in\! T$). Therefore $\SC(G) > C_{RS} + 6k'(3s-2t) -2k' + 4k > C $
\item We have already shown that $G$ is a spanning tree of $GM$, thus $G$ is \pNe{} topology, and $\SC(G) - C_{TT}(G) = C - C_{TT}$ (where $C_{TT}(G)$ is the cost arising for the the traffic from $T$ to $T$ in $G$). In order for the $C_{TT}(G) \leq C_{TT}$, let $s_i$ be the number of $\sigma \!\in\! S$ with degree $i$, we have $C_{TT}(G) =2 (6s_4 + 10s_3 + 12s_2 + 12s_1)$ with $3s_4 + 2s_3 + s_2 = 3t$ and $s_1 + s_2 + s_3 + s_4 = s$. Manipulating expressions:
\begin{align*}
s_4 =\ & t - \frac{2s_3 + s_2}{3} \qquad\qquad\qquad\qquad\hspace{1mm} s_1 = s - t + \frac{2s_3 + s_2}{3} - s_3 - s_2 \\
C_{TT}(G) =\ &2( (6t - 4s_3 - 2s_2) + 10s_3 + 12s_2 + (12s - 12t + 8s_3 + 4s_2 - 12s_3 - 12s_2) )\\
C_{TT}(G) =\ &2( 12s - 6t + 2s_3 + 2s_2 ) \quad C_{TT}(G) = 2( 2(6s - 3t) + 2s_3 + 2s_2)
\end{align*}
Where $s_3 = s_2 = 0$ then $s_4 = t$ and $s_0 = s$ and $C_{TT}(G) = C_{TT}$. Otherwise, $s_4 < t$ and $s_3 + s_2 > 0$, hence $C_{TT}(G) > C_{TT}$.
\end{enumerate}
When $G$ is \pNe{} topology and $\SC(G) \leq C$ the $\sigma \!\in\! V$ nodes such that $\delta(\sigma;G) = 4$ form form an exact 3-cover to $T$.
\qed
\end{proof}

\begin{proof}[theorem~\ref{thm:com}]
We transform \rxtc{} to \lsc{}. Let $\langle T, S \rangle$ be a \rxtc{} instance, we construct a $\langle M, \pi, (c_i)_{i \in V(GM)}, C \rangle$ \lsc{} instance as follows:

{\newsavebox{\tempboxx}
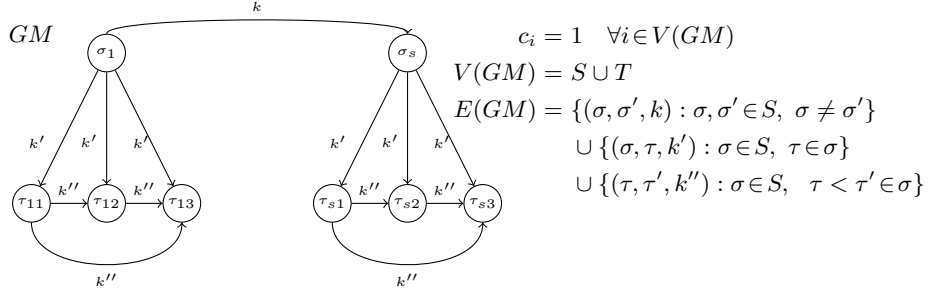
\begin{center}
\begin{figure}[!ht]
\sbox{\tempboxx}{
\begin{tikzpicture}

 \node at (0,2.25) {$GM$};

 \node[draw=black,circle,inner sep=0pt, minimum size=5mm] at (1,2) (s1) {\tiny{$\sigma_1$}};
 \node[draw=black,circle,inner sep=0pt, minimum size=5mm] at (5,2) (ss) {\tiny{$\sigma_s$}};

 \node[draw=black,circle,inner sep=0pt, minimum size=5mm] at (0,0) (t11) {\tiny{$\tau_{11}$}};
 \node[draw=black,circle,inner sep=0pt, minimum size=5mm] at (1,0) (t12) {\tiny{$\tau_{12}$}};
 \node[draw=black,circle,inner sep=0pt, minimum size=5mm] at (2,0) (t13) {\tiny{$\tau_{13}$}};

 \node[draw=black,circle,inner sep=0pt, minimum size=5mm] at (4,0) (ts1) {\tiny{$\tau_{s1}$}};
 \node[draw=black,circle,inner sep=0pt, minimum size=5mm] at (5,0) (ts2) {\tiny{$\tau_{s2}$}};
 \node[draw=black,circle,inner sep=0pt, minimum size=5mm] at (6,0) (ts3) {\tiny{$\tau_{s3}$}};

 \draw [->] (s1) -- node[near end, above left] {\tiny{$k'$}} (t11);
 \draw [->] (s1) -- node[near end, above left] {\tiny{$k'$}} (t12);
 \draw [->] (s1) -- node[near end, above left] {\tiny{$k'$}} (t13);

 \draw [->] (ss) -- node[near end, above left] {\tiny{$k'$}} (ts1);
 \draw [->] (ss) -- node[near end, above left] {\tiny{$k'$}} (ts2);
 \draw [->] (ss) -- node[near end, above left] {\tiny{$k'$}} (ts3);

 \draw [->] (s1) .. controls +(up:0.5cm) and +(up:0.5cm) .. node[above] {\tiny{$k$}} (ss) ;

 \draw [->] (t11) -- node[above] {\tiny{$k''$}} (t12);
 \draw [->] (t12) -- node[above] {\tiny{$k''$}} (t13);
 \draw [->] (t11) .. controls +(down:1cm) and +(down:1cm) .. node[below] {\tiny{$k''$}} (t13);

 \draw [->] (ts1) -- node[above] {\tiny{$k''$}} (ts2);
 \draw [->] (ts2) -- node[above] {\tiny{$k''$}} (ts3);
 \draw [->] (ts1) .. controls +(down:1cm) and +(down:1cm) .. node[below] {\tiny{$k''$}} (ts3);
\end{tikzpicture}
}
\subfloat{\usebox{\tempboxx}}%
\hspace{-3.8cm}
\subfloat{\vbox to \ht\tempboxx{%
  \vfil
\begin{flushleft}
{\small
\begin{align*} 
c_i    =\ & 1 \quad \forall i \!\in\! V(GM) \\
V(GM) =\ & S \cup T \\
E(GM) =\ & \{(\sigma,\sigma',k) : \sigma,\sigma'\!\in\!S,\ \sigma \neq \sigma'\} \\
          & \cup \{(\sigma,\tau,k') : \sigma\!\in\!S,\ \tau\!\in\!\sigma\} \\
          & \cup \{(\tau,\tau',k'') : \sigma\!\in\!S,\ \ \tau < \tau'\!\in\!\sigma\} \\ \\
\end{align*}
}
\end{flushleft}
  \vfil}\label{fig:complexitat:optim}}%
\end{figure}
\end{center}}

Given a possible 3-cover $S' \subseteq S$ for $\langle T, S \rangle$ we can construct the associated graph $G(S')$ where $V(G(S')) = V(GM)$ and
\[E(G(S')) = \{ (r,\sigma) : \sigma \!\in\! S \} \cup \{ (\sigma,\tau) : \sigma \!\in\! S', \tau \!\in\! \sigma\}\]
If $S'$ is an exact 3-cover then $G(S')$ is a spanning tree of $GM$. We define $C$ in such way that if $S'$ is a exact 3-cover then $C(G(S')) = C$.

Let $k > \pi$ and let $G$ be the optimal topology. The induced graph by the $S$ vertices, $G[S]$, is a $s$-complete graph: if there is a pair $\sigma,\sigma'\!\in\!S$ such that $dist(\sigma,\sigma';G) > 1$, then the graph $G' = G + \sigma\sigma'$ has social cost strictly less, $SC(G') < SC(G)$. We assign a value for $k$ enough large so that any topology $G$, if $G[S] \neq K_s$ then $SC(G) > C$. 

In $G(S')$ we have that for $\sigma\!\in\!S'$ the distance between $\sigma$ and $\tau \in \sigma$ is $1$, for $\sigma\!\notin\!S'$ this distance is $2$. If $\pi > k'$, in the optimal topology the $\tau\!\in\!T$ nodes are connected only to one $\sigma$: since $G[S] = K_s$ the $\tau$ nodes are at a distance $2$ of all $\sigma$, and there is no social benefit in adding links in order to reduces these distances. We assign a value for $k'$ enough large so that if $(\tau,\sigma)\!\in\!G$ where $\tau\!\notin\!\sigma$ then $SC(G) > C$.

Finally for the traffic between $\tau,\tau'\!\in\!\sigma$ we require that only the topologies $B = G(S')$ for an exact 3-cover $S'$ satisfies $SC(G) \leq C$.

Let $S'$ be an exact 3-cover, we denote by $C_{XY}$ the routing cost arising for the traffic sent from $X$ to $Y$, in the $G(S')$ topology.
\[
C_{SS} = k(s^2-s) \qquad C_{ST} = 6k'(2s-t) \qquad C_{TT} = 6k''(3s-t)
\]

We define the $k$, $k'$, $k''$ and $\pi$ as follows:
\begin{align*}
& \begin{array}{l l}
k = C_{ST} + C_{TT} + \pi(6t+s^2)         & k' =  C_{TT} \\
k'' = 1    & \pi = C_{ST} + C_{TT} \\
\end{array} \\
& C = C_{SS} + C_{ST} + C_{TT} + \pi(6t+s^2-s)
\end{align*}

Now we show that $\langle T,S\rangle \!\in\! \rxtc \iff \langle M, \pi, (c_i)_{i \in V(GM)}, C \rangle \!\in\! \lsc$:

If $S' \subseteq S$ is an exact 3-cover of $T$, the tree $G(S')$ has  $SC(G(S')) = C$.

If $\langle  M, \pi, (c_i)_{i \in V(GM)}, C  \rangle\!\in\!\lsc{}$ exists a graph $G$ such that $SC(G) \leq C$. $G$ is connected and satisfies:
\begin{enumerate}[label=(\roman{*})]
\item $\forall\sigma\sigma'\!\in\!S\ (\sigma,\sigma')\!\in\!E$, the traffic sent from $S$ to $S$ entails a cost $C_{ST}$.
Otherwise let's consider the cost arising for the traffic sent from $S$ to $S$, if the induced graph $G[S]$ is not the $s$-complete, we can add a link between two $\sigma,\sigma'\!\notin\!E$, so the routing cost has a strictly decrease, greater or equal than $k-\pi$. Graphs $G'$ that minimizes the traffic from $S$ to $S$ have $G'[S] = k_s$, and the routing cost from $S$ to $S$ is $k(s^2-s) = C_{SS}$. If $G[S]$ is not the $K_s$ we have $SC(G) > k(s^2-s) + k - \pi = C_{SS} + C_{ST} + C_{TT} + \pi(6t+s^2 - 1) \geq C$.

\item $|E(G)| = \frac{(s^2-s)}{2}+3t$. Otherwise or $G$ is not connected or $SC(G) > C_{SS} + \pi(6t+s^2-s) + \pi = C_{SS} + \pi(6t+s^2-s) + C_{ST} + C_{TT} = C$.

\item $\forall\tau\!\in\!T\ \exists\sigma\!\in\!S,\ \tau\!\in\!\sigma\land(\sigma,\tau)\!\in\!E$. Otherwise there is a partition $\{X,Y,Z\}$ of $T$, where
$\tau\!\in\!X \!\!\!\iff\!\!\! \exists\sigma\!\in\!S,\ \tau\!\in\!\sigma\land(\sigma,\tau)\!\in\!E$,
$\tau\!\in\!Y \!\!\!\iff\!\!\! \exists\sigma\!\in\!S,\ \tau\!\notin\!\sigma\land(\sigma,\tau)\!\in\!E$,
$\tau\!\in\!Z \!\!\!\iff\!\!\! \forall\sigma\!\in\!S,\ (\sigma,\tau)\!\notin\!E$, such that $|Y| + |Z| > 0$ (notice that the $\tau\!\in\!Z$ nodes are connected to others $\tau'\!\in\!T$).
Then $SC(G) > C_{SS} + \pi(6t+s^2-s) + 4k'(3s) - |X+Z|2k' + |Z|2k' \geq C_{SS} + \pi(6t+s^2-s) + C_{ST} + 2k' > C$.

\item $\forall (\sigma,\tau)\!\in\!E,\ \delta(\sigma;G) = 4$. Now we have that $\SC(G) - C_{TT}(G) = C - C_{TT}$ (where $C_{TT}(G)$ is the routing cost arising for the traffic between $T$ and $T$). For $C_{TT}(G) \leq C_{TT}$, let $s_i$ be the number of $\sigma \!\in\! S$ with grade $i$, we have $C_{TT}(G) = 2 (6s_4 + 8s_3 + 9s_2 + 9s_1)$ with $3s_4 + 2s_3 + s_2 = 3t$ and $s_1 + s_2 + s_3 + s_4 = s$. Manipulating expressions:
\begin{align*}
s_4 =& \ t - \frac{2s_3 + s_2}{3} \qquad s_1 = s - t + \frac{2s_3 + s_2}{3} - s_3 - s_2 \\
C_{TT}(G) =& \ 2((6t-4s_3-2s_2) + 8s_3 + 9s_2 + (9s - 9t + 6s_3 + 3s_2 - 9s_3 - 9s_2)) \\
C_{TT}(G) =& \ 2(9s-3t+s_2+s_3) \qquad C_{TT}(G) = 2(3(3s-t) + s_2 + s_1)
\end{align*}
Where $s_3 = s_2 = 0$ then $s_4 = t$, $s_0 = s$ and $C_{TT}(G) = C_{TT}$.
Where $s_4 < t$ then $s_3 + s_2 > 0$ thus $C_{TT}(G) > C_{TT}$.
\end{enumerate}
\qed
\end{proof}

\end{document}